\documentclass[review,1p]{elsarticle}

\usepackage{amssymb}

\usepackage{tipa}

\usepackage{graphicx}
\DeclareGraphicsExtensions{.pdf,.jpeg,.png,.jpg}
\usepackage{amsmath}
\usepackage{amssymb}
\usepackage{amsthm}
\usepackage{mathabx} 
\mathchardef\mhyphen="2D 

\usepackage[ruled,vlined,linesnumbered]{algorithm2e}
\SetAlgoSkip{medskip}

\usepackage{array}

\usepackage[small,bf,hang]{caption}
\usepackage[caption=false,font=footnotesize]{subfig}
\usepackage{fixltx2e}
\usepackage{stfloats}
\usepackage{floatflt}
\usepackage{url}
\usepackage[active]{srcltx}
\usepackage{booktabs}
\usepackage{framed}
\usepackage{verbatim}
\usepackage[stable]{footmisc}
\usepackage{paralist}
\usepackage{xspace} 
\usepackage{latexsym}

\usepackage{stmaryrd}

\newtheorem{theorem}{Theorem}
\newtheorem{lemma}[theorem]{Lemma}

\newtheorem{corollary}[theorem]{Corollary}
\newtheorem{proposition}[theorem]{Proposition}

\newenvironment{definition}[1][Definition]{\begin{trivlist}
\item[\hskip \labelsep {\bfseries #1}]}{\end{trivlist}}

\usepackage{xspace} 

\def\eg{\emph{e.g.,}\xspace}
\def\floor#1{\lfloor #1 \rfloor}
\def\ceil#1{\lceil #1 \rceil}
\def\combp{{\sc CombPart}\xspace}

\def\mark{{\sc MarkForbidden}\xspace}
\def\splitp{{\sc SplitPart}\xspace}
\def\splitm{{\sc SplitMark}\xspace}
\def\simp{{\sc SimplePart}\xspace}
\def\plp{{\sc LpPart}\xspace}
\def\pilp{{\sc IlpPart}\xspace}

\def\alg{{\sc Alg}\xspace}
\def\yes{{\sc Yes}\xspace}
\def\no{{\sc No}\xspace}
\def\ie{\emph{i.e.,}\xspace}
\def\lmn{\mathtt{Lmn}}
\def\Rnf{\mathtt{Rnf}}
\def\rnf{\mathtt{rnf}}
\def\Ldist{\mathtt{Ldist}}
\def\ldist{\mathtt{ldist}}
\def\bg{\mathtt{begin}}
\def\nd{\mathtt{end}}
\def\rit{\mathtt{right}}
\def\size{\mathtt{size}}
\def\prv{\mathtt{prev}}
\def\nxt{\mathtt{next}}
\def\inlay{\mathtt{Inlay}}
\def\ptr{{\mathtt{\shortrightarrow}}}

\usepackage{calc,tikz}
\usetikzlibrary{matrix,chains,scopes,positioning,arrows,fit}
\usetikzlibrary{calc,shapes.geometric,shapes.arrows}
\usetikzlibrary{decorations,decorations.pathmorphing}
\usetikzlibrary{decorations.pathreplacing}
\usetikzlibrary{patterns}
\usetikzlibrary{shapes.symbols}

\pgfdeclarelayer{background}
\pgfdeclarelayer{foreground}
\pgfsetlayers{background,main,foreground}

\tikzstyle{vertex}=[thick, inner sep=0pt, minimum size=8mm]
\tikzstyle{allow}=[vertex, circle, draw=black]
\tikzstyle{forbid}=[vertex, forbidden sign, draw=black]
\tikzstyle{uncoloured}=[vertex, circle, draw=black!80, fill=black!40]
\tikzstyle{new forbid}=[forbid, draw=red!90]
\tikzstyle{old forbid}=[forbid, draw=red!50]
\tikzstyle{block}=[rectangle, fill=gray!40, inner sep=5pt, rounded corners=2mm]
\tikzstyle{clique}=[ellipse, scale=0.75, line width=0.2ex, inner sep=0pt, minimum height=3cm]
\tikzstyle{active clique}=[clique, draw=black!60]
\tikzstyle{inactive clique}=[clique, draw=black!30]
\tikzstyle{bar}=[draw, dotted, line width=0.1ex]

\journal{Discrete Applied Mathematics}

\begin{document}

\begin{frontmatter}



\title{Component Coloring of Proper Interval and Split Graphs}
\author{Ajit~Diwan}
\ead{aad@cse.iitb.ac.in}
\author{Soumitra~Pal}
\ead{mitra@cse.iitb.ac.in}
\author{Abhiram~Ranade}
\ead{ranade@cse.iitb.ac.in}
\address{Department of Computer Science and Engineering, \\
 Indian Institute of Technology Bombay, \\
 Powai, Mumbai 400076, India.}


\begin{abstract}
 We introduce a generalization of the well known graph (vertex) coloring
 problem, which we call the problem of \emph{component coloring
 of graphs}. Given a graph, the problem is to color the vertices
 using the minimum number of colors so that the size of each connected
 component of the subgraph induced by the vertices of the same color
 does not exceed~$C$. We give a linear time algorithm for the problem on
 proper interval graphs. We extend this algorithm to solve two weighted
 versions of the problem in which vertices have integer weights. In the
 \emph{splittable} version the weights of vertices can be split into
 differently colored parts, however, the total weight of a monochromatic
 component cannot exceed~$C$. For this problem on proper interval graphs
 we give a polynomial time algorithm. In the \emph{non-splittable} version
 the vertices cannot be split. Using the algorithm for the splittable
 version we give a $2$-approximation algorithm for the non-splittable
 problem on proper interval graphs which is NP-hard. We also prove that
 even the unweighted version of the problem is NP-hard for split graphs.
\end{abstract}


\end{frontmatter}

\section{Introduction}

The vertex coloring problem is to color the vertices of a graph using the
minimum number of colors so that no two adjacent vertices are assigned
the same color. In this paper, we introduce and study a generalization
of the vertex coloring problem. In this generalized problem, called the
problem of \emph{component coloring of graphs}, we allow two adjacent
vertices to be assigned the same color. It is customary to consider
two variations: {\em unweighted} and {\em weighted}. In the unweighted
version of the problem, given an graph $G=(V,E)$, the objective is to
color the vertices using the minimum number of colors such that the
size of any monochromatic component, \ie the connected component of the
subgraph induced by the vertices of the same color, does not exceed~C.
The vertex coloring problem is a special case of the unweighted component
coloring problem where $C=1$, and each monochromatic component consists
of a single vertex.

In the weighted version of the problem, given an graph $G=(V,E)$ and
for each $v \in V$ a rational weight $W(v) \in (0, 1]$, the objective
is to color the vertices using the minimum number of colors such that
the total weight of any monochromatic component, does not exceed~1.

Since the vertex coloring problem is NP-hard on general
graphs~\cite{garey1979cig}, the unweighted (and hence weighted) component
coloring problem is also NP-hard on general graphs.

Our formulation of the component coloring problem is motivated by a
problem on scheduling transmission requests on \emph{light-trails}, a
hardware solution for bandwidth provisioning in optical WDM (Wavelength
Division Multiplexing) networks~\cite{chlamtac2003light}. In a path
network of processors using light-trails, each processor has an optical
shutter for each wavelength which can be configured to be switched
ON/OFF for allowing/blocking the light signal pass through it. For each
wavelength, by suitably configuring the optical shutter at each processor,
the logical path network can be partitioned into subpath networks
in which multiple transmissions can happen in parallel, provided the
total bandwidth requirement of the transmissions assigned to a subpath
does not exceed the capacity of a wavelength. Such subpaths, in which
only the end processors have their optical shutters blocked, are called
\emph{light-trails}. A light-trail can serve only the transmissions having
both source and destination within the light-trail. If a transmission
is assigned to a light-trail, it uses the complete physical span of the
light-trail.  Given a set of transmission requests, each with a bandwidth
requirement, the scheduling problem is to configure the optical shutters
at the processors so that the minimum number of wavelengths is required by
the light-trails to serve all transmission requests.

A graph $G$ is an \emph{interval graph} if there exists a family
$\mathcal{I}$ of intervals in a linearly ordered set (like the real line),
and there exists a one-to-one correspondence between the vertices of $G$
and the intervals in $\mathcal{I}$ such that two vertices are adjacent
if and only if the corresponding intervals intersect. If no interval
of $\mathcal{I}$ properly contains another, set theoretically, then $G$
is called a \emph{proper interval graph}.

The light-trail scheduling problem on path networks can be posed as
a component coloring problem on interval graphs as follows. For each
transmission request, create a vertex with weight equal to the bandwidth
requirement, expressed as a fraction of the wavelength capacity. Two
vertices are adjacent if the corresponding transmissions overlap, \ie
they use at least one common link. Given a solution to the component
coloring problem, a solution to the light-trail scheduling problem can
be constructed as follows. For each of the used colors, use a separate
wavelength. For each wavelength, construct a separate light-trail for
each monochromatic component of the corresponding color. Note that
the light-trails on a wavelength do not intersect with each other. All
transmission requests corresponding to the vertices of a monochromatic
component are served by the corresponding light-trail. The physical span
of the light-trail is the union of the physical spans of all requests
in it. For each wavelength, the optical shutters in only the processors
at the endpoints of all light-trails on the corresponding wavelength are
configured to be OFF; optical shutters of other processors are configured
to be ON.

As mentioned in~\cite{pal2009slt}, the light-trail scheduling
problem is similar to the problem of scheduling in reconfigurable bus
architectures~\cite{elgindy1996rmb, DBLP:journals/ijcsa/WankarA09},
and hence component coloring applies there too.

The unweighted component coloring problem for $C=1$, \ie the
vertex coloring problem, has a polynomial time algorithm on interval
graphs~\cite{olariu1991optimal}.  However, the complexity of the problem
on interval graphs for general $C$ is not known. In this paper we give
a polynomial time algorithm for the problem on proper interval graphs
for general $C$. Since the problem arises in scheduling light-trails on
path networks, we assume that an interval representation of the graph
is also available. Our first result is the following.

\begin{theorem}
 Given a proper interval graph $G=(V,E)$ with an interval representation,
 there exists an algorithm that solves the unweighted component coloring
 problem on $G$ in $O(|V|)$ time.  \label{thm:pig}
\end{theorem}

We also consider a \emph{splittable weighted} version of the
component coloring problem in which each vertex of the input graph
has an integer weight which can be divided among multiple copies
of the vertex and these copies can be colored separately. However,
the total weight of a monochromatic component in the resultant graph
should not exceed~$C$. Again, this is motivated by a variation of the
light-trail scheduling problem in which the bandwidth requirement of a
transmission can be divided into multiple transmissions between the same
source-destination pair.   We extend the algorithm for the unweighted
problem to solve this splittable weighted problem on proper interval
graphs. So our second result is the following.

\begin{theorem}
 Given a proper interval graph $G=(V,E)$ with an interval representation,
 there exists an algorithm that solves the splittable weighted component
 coloring problem on $G$ in $O(|V|^2)$ time. \label{thm:splittable}
\end{theorem}

However, the (\emph{non-splittable}) weighted version of the problem is
NP-hard even on proper interval graphs. This comes from the fact that
the complete graph $K_n$ is a proper interval graph and the weighted
component coloring problem on $K_n$ is an instance of NP-hard Bin
Packing problem~\cite{garey1979cig}. We use the algorithm for the
splittable weighted problem to get a $2$-approximation algorithm for
the non-splittable weighted problem on proper interval graphs.

\begin{theorem}
 Given a proper interval graph $G$ with an interval representation,
 there exists a $2$-approximation algorithm for the non-splittable
 weighted component coloring problem on $G$. \label{thm:weighted}
\end{theorem}

The vertex coloring problem also has a polynomial time algorithm for split
graphs, \ie when the vertex set can be partitioned into an independent
set and a clique~\cite{golumbic2004algorithmic}. However, for general
$C$, we prove that the unweighted component coloring is NP-hard  on
split graphs. So our final result is the following.

\begin{theorem}
 The component coloring problem is NP-hard for split graphs.
 \label{thm:split}
\end{theorem}

The rest of the paper is organized as follows. We begin in
Section~\ref{sec:prev_work} by comparing our work with previous related
work. In Section~\ref{sec:prelim} we present some pertinent definitions
and known results. In Section~\ref{sec:eqv_chordal} we show that for the
class of chordal graphs, the component coloring problem is equivalent
to a vertex partitioning problem. Note that the interval graphs and the
split graphs are chordal. We show in Section~\ref{sec:eqv_pig} that for
the class of proper interval graphs, it is enough to solve a simpler
version of the partitioning problem which we call the block-partitioning
problem. We give an LP based algorithm for the block-partitioning problem
in Section~\ref{sec:lp}.  We give a combinatorial algorithm for the same
problem in Section~\ref{sec:comb}. In Section~\ref{sec:splittable} we
extend this algorithm to solve the splittable weighted problem. Based
on the algorithm for the splittable weighted problem we give a
$2$-approximation algorithm for the non-splittable weighted problem in
Section~\ref{sec:weighted}. We prove the NP-hardness of the problem on
split graphs in Section~\ref{sec:split}.


\section{Previous Work}
\label{sec:prev_work}

In the graph coloring literature, there are
papers~\cite{edwards2005monochromatic, linial2008graph} to solve a
problem that is a kind of dual to the unweighted component coloring
problem. Here, the objective is to minimize the size of the largest
monochromatic component in a coloring using a fixed number of colors. The
paper~\cite{edwards2005monochromatic} shows that for a $n$-vertex graph
of maximum degree 4, there exists an algorithm that uses 2~colors and
produces a coloring in which the size of the largest monochromatic
component is $O(2^{(2\log_2n)^{1/2}})$. For a family of minor-closed
graphs, the paper~\cite{linial2008graph} shows that if $\lambda$ colors
are used, the size of the largest monochromatic component is in between
$\Omega(n^{2/(2\lambda-1)})$ and $O(n^{2/(\lambda+1)})$ for every fixed
$\lambda$. However, in our knowledge, there is no work in the graph
coloring literature for the versions of the problem we formulated.

The NP-hard light-trail scheduling problem with arbitrary
bandwidth requirements on ring networks and general networks
has generally been solved using heuristics and evaluated
experimentally~\cite{balasubramanian2005ltn, ayad2007eoa, gumaste2007hao,
wu2006opn03, luo2009integrated, gokhale2010cloud2} without any bound on
the performance. For path/ring networks, the paper~\cite{pal2009slt} gives
an approximation algorithm that uses $O(\omega + \log p)$ wavelengths
where $p$ is the number of processors in the network and $\omega$ is
the \emph{congestion}, \ie the maximum total traffic required to pass
through any link. For the corresponding component coloring problem,
$p$ is the number of distinct end points of the intervals in the given
interval representation, and $\omega$ is the weight of a maximum clique
and hence a lower bound on the number of colors used. Thus the algorithm
in~\cite{pal2009slt} is a constant factor approximation algorithm
with an additive term $\log{p}$ for the component coloring problem on
interval/circular-arc graphs. Note that in general $p \ll 2n$.


\section{Preliminaries}
\label{sec:prelim}

Throughout this paper, let $G=(V,E) = (V(G), E(G))$ be a simple,
undirected graph and let $n=|V|$ and $m=|E|$. We also assume that $G$
is connected. If $G$ is not connected, the results in this paper can
be applied separately to each of its connected components.  The set of
vertices adjacent to a vertex $v \in V$ is represented as $N(v)$. For a
set $S \subseteq V$, the sub-graph of $G$ \emph{induced by} $S$ is $G[S]
= (S, E(S))$ where $E(S)=\{(u,v) \in E \mid u, v \in S\}$. For a set $S
\subset V, V-S$ denotes $G[V \setminus S]$. A \emph{clique} of $G$ is a
set of pair wise adjacent vertices of $G$. The \emph{size of a clique}
is the number of vertices in it. A \emph{maximal clique} is a clique of
$G$ that is not properly contained in any clique of $G$. A \emph{maximum
clique} is a clique of maximum size. The \emph{clique number} of $G$,
denoted by $\omega(G)$ or simply $\omega$, is the size of a maximum
clique of $G$.  An \emph{independent set} of $G$ is a set of pairwise
non-adjacent vertices in it.

A weighted graph $G=(V,E,W)$ has a weight $W(v) \in  \mathbb{Z}_{\ge 0}$
associated with each vertex $v \in V$. The \emph{weight-split graph}
of a weighted graph $G=(V,E,W)$, in short $WSP(G)$, is the weighted
graph $G'=(V',E',W')$ such that the weight of each $v \in V$ is divided
among a separate set of vertices $v_1, \ldots, v_{n_v}$ in  $V'$,
\ie $\sum_{j=1}^{n_v} W'(v_j) = W(v)$ and for each edge $(u,v) \in
E$ there is an edge $(u_i, v_j) \in E'$ for all $i=1,\ldots,n_u$ and
$j=1,\ldots,n_v$. The \emph{weight-expanded graph} of a weighted graph
$G=(V,E,W)$, in short $WXP(G)$, is the unweighted graph $G'=(V',E')$ such
that if we put a weight $1$ to each vertex in $V'$ then the resulting
weighted graph is a weight-split graph of $G$. We will use the following
notations: $n'=|V'|$ and $m'=|E'|$.

Coloring of a graph is an assignment of colors to its vertices. A
\emph{$\lambda$-assignment} of a graph $G=(V,E)$ is a map from $V$
to some set of $\lambda$ colors such as $\{1, \ldots, \lambda\}$;
this assignment may not be `proper' in the standard notion of graph
(vertex) coloring that two adjacent vertices must be assigned different
colors. A \emph{color class}~$i$ is the set of vertices assigned color $i$
under the $\lambda$-assignment. A \emph{monochromatic component} of $G$
under a $\lambda$-assignment is a component of the sub-graph induced by a
single color class, or in other words, a maximal connected monochromatic
sub-graph. Following the terminology of~\cite{edwards2005monochromatic},
we call a monochromatic component a \emph{chromon}. The \emph{size of
a chromon} is the number of vertices in it. For a weighted graph, the
\emph{weight of a chromon} is the sum of weights of vertices in it.

An unweighted (weighted) graph is \emph{$[\lambda, C]$-colorable} if it
has a $\lambda$-assignment in which every chromon has size (weight)
at most~$C$ and such an assignment is called a \emph{$[\lambda,
C]$-coloring}. A \emph{$C$-component coloring} of graph $G$ is a
$[\lambda, C]$-coloring with the minimum~$\lambda$. Sometimes we will
simply refer to the problem of finding a $C$-component coloring of a
graph as the \emph{coloring} problem.

A weighted graph $G$ is \emph{$[\lambda, C]$-split colorable} if it
has a weight-split graph $G'$ which is \emph{$[\lambda, C]$-colorable}
and such a coloring is called a \emph{$[\lambda, C]$-split coloring} of
$G$. A \emph{$C$-split component coloring} of graph $G$ is a $[\lambda,
C]$-split coloring with the minimum~$\lambda$.  Sometimes we will simply
refer to the problem of finding a $C$-split component coloring of a
graph as the \emph{split coloring} problem.

The component coloring problem can be seen as solving two problems
simultaneously, (i) partitioning the vertex set into chromons and (ii)
assigning colors to the chromons. The partitioning should be such that
if each part is contracted to a single vertex, the resulting graph can
be colored using as few colors as possible. Since the size of a maximum
clique in the contracted graph plays a major role in determining the
number of colors used, at least for some graphs classes such as perfect
graphs, we have to ensure that the cliques in the original graph does
not intersect too many parts. We formally define the partitioning problem
as follows:
\begin{definition}
A graph $G=(V,E)$ is said to have a \emph{$[\lambda,C]$-partition}
if and only if there is a partition $\Pi=\{P_1, P_2, \ldots, P_t\}$
of $V$, $P_i \subseteq V$, $P_i \cap P_j = \emptyset$ for all $i \ne j$
such that the following constraints are satisfied:
\begin{itemize}
 \item \emph{connectedness} -- the subgraph induced by each part $P_i$,
  \ie $G[P_i]$ is connected,
 \item \emph{size} -- each part $P_i$ has at most $C$ vertices, and
 \item \emph{clique intersection} -- any clique in $G$ intersects at most
  $\lambda$ parts ($\lambda$ will subsequently be called the
  \emph{clique intersection} of the partition).
\end{itemize}
\end{definition}
A \emph{$C$-component partition} of a graph is a $[\lambda,C]$-partition
with the minimum~$\lambda$. We will refer to the problem of finding a
$C$-component partition as the \emph{partition} problem.

We study the coloring problem on interval graphs and split graphs. Each of
these classes of graphs is a subclass of the class of chordal graphs. A
graph is \emph{chordal} if each of its cycles of four or more vertices
has a \emph{chord}, which is an edge joining two vertices that are not
adjacent in the cycle. There are many characterizations of chordal graphs
(see~\cite{golumbic2004algorithmic} for more details). We will use the
characterization of a chordal graph based on \emph{perfect elimination
ordering} or, in short, \emph{PEO}.  A vertex $v$ of $G$ is called
\emph{simplicial} if its neighbors $N(v)$ form a clique. An ordering
$\sigma=[v_1,v_2,\ldots,v_n]$ of vertices is a PEO if each vertex $v_i$
is a simplicial vertex of the induced subgraph $G[v_i,\ldots,v_n]$.

\begin{proposition}[\cite{golumbic2004algorithmic}]
Let $G=(V,E)$ be an undirected graph. Then $G$ is a chordal graph if and 
only if $G$ has a PEO. Moreover, any simplicial vertex can start a PEO.
\label{pro:chordal}
\end{proposition}

A graph $G=(V,E)$ is a \emph{split graph} if there is a partition
$V=S+Q$ of its vertex set into an independent set $S$ and a clique
$Q$. There is no restriction on edges between vertices of $S$ and
$Q$. A graph $G=(V,E)$ is an \emph{interval graph} if there exists a
family $\mathcal{I}=\{I_v \mid v \in V\}$ of intervals on a real line
such that for distinct vertices $u,v$ in $G$, $(u,v) \in E$ if and
only if $I_u \cap I_v \ne \emptyset$. Such a family $\mathcal{I}$ of
intervals is commonly referred to as the \emph{interval representation}
of $G$. Given an interval representation of $G$, consider a cycle of
more than $3$ vertices, and the corresponding intervals in ascending left
endpoints. Since the rightmost interval intersects the leftmost interval,
it also intersects the intervals in between them. Hence $G$ is also
chordal. It will be convenient to let $Left(I_v)$ and $Right(I_v)$ stand
for the left and right endpoint of the interval $I_v$, respectively. The
family $\mathcal{I}$ is the interval representation of a \emph{proper
interval graph (PIG)} if and only if no interval is properly contained
in another. Interval graphs and split graphs are easily seen to be
chordal~\cite{golumbic2004algorithmic}.

\begin{proposition}[\cite{olariu1991optimal}]
There exists an $O(m+n)$ time algorithm to get an interval representation
of a given interval graph. \qed \label{pro:intorder}
\end{proposition}

However, since the component coloring problem is motivated by the
light-trail scheduling problem, in this paper we will assume that an
interval representation $\mathcal{I}=\{I_v \mid v \in V\}$ is given for
the input PIG $G=(V,E)$.

Now consider the linear order~$\prec$ on $V$ defined as
follows. For $u,v \in V$, $u \prec v$ if and only if $Left(I_u) <
Left(I_v)$ or $\lgroup(Left(I_u) = Left(I_v))$ and $(Right(I_u) \le
Right(I_v))\rgroup$. We call this ordering $v_1 \prec v_2 \prec \cdots
\prec v_n$ the \emph{canonical ordering}. In the rest of the paper, we
use numbers $1$ to $n$ to represent the vertices where $i$ represents
the vertex that appears $i$th in the canonical ordering. Hence, $v$
will be interchangeably used to represent a vertex $v \in V$ as well as
its position in the canonical ordering. If $u \prec v$ then $u$ is said
to be on the left of $v$ and $v$ is said to be on the right of $u$.

\begin{proposition}[\cite{olariu1991optimal}]
 A graph $G=(V,E)$ is an interval graph if and only if there exists a
 linear order $\prec$ on $V$ such that for every choice of vertices
 $u, v, w$ with $u \prec v \prec w$, $(u,w) \in E$ implies $(u,v)
 \in E$. \qed \label{pro:intvorder}
\end{proposition}

For PIGs the canonical ordering not only satisfies the conditions in
Proposition~\ref{pro:intvorder} but, in fact, satisfies a stronger
property:

\begin{proposition}[``The Umbrella Property'' \cite{looges1993optimal}]
 A graph $G=(V,E)$ is a PIG, if and only if, there exists a linear order
 $\prec$ on $V$ such that for every choice of vertices $u, v, w$,
 with $u \prec v \prec w$, $(u,w) \in E$ implies both $(u,v) \in E$
 and $(v,w) \in E$. \qed \label{pro:propervorder}
\end{proposition}

A \emph{block}\footnote{Some authors use the term block to represent
what we call a clique.} in a PIG a is a set of vertices which are
consecutive in the canonical ordering. We will represent a block starting
at a vertex $u$ and ending at a vertex $v$ as the interval $[u,v]$. An
immediate corollary of Proposition~\ref{pro:propervorder} is that every
edge $(u,v) \in E$ induces a clique $[u,v]$. Also, any maximal clique of a
PIG can be represented by a single edge between the two end vertices,
say $(u,v)$, or by the block $[u,v]$.

\begin{corollary}
 Let $S$ be a connected subgraph of a PIG and $v_1, v_2, \ldots,
 v_t$ be the vertices of $S$ arranged in the canonical ordering. Then
 there must be an edge $(v_i, v_{i+1})$ for all $i=1,\ldots,
 t-1$. \label{cor:edge_seq}
\end{corollary}

\begin{proof}
 Consider the two vertices $v_i$ and $v_{i+1}$. Since $S$ is connected
 there must be an edge $(v_j, v_k)$ where $j \le i$ and $i+1 \le k$. Then
 $[v_j,v_k]$ is a clique. Thus there is an edge $(v_i, v_{i+1})$.
\end{proof}

\begin{proposition}
 If an interval representation is given for a PIG $G$, then the
 vertices of $G$ can be arranged in canonical ordering in $O(n)$ time.
 \label{pro:algocanon}
\end{proposition}

\begin{proof}
Since there are at most $t \le 2n$ endpoints of all intervals in the
representation, the intervals can be sorted in canonical ordering using
bucket sort in $O(n)$ time.
\end{proof}

\begin{proposition}
 If an interval representation is given for a PIG $G$, then the maximal
 cliques of $G$ can be found in $O(n)$ time.  \label{pro:algomaxq}
\end{proposition}

\begin{proof}
Let $\mathcal{I} = \{I_v \mid v \in V\}$ be an interval representation
of $G$.  Without loss of generality we assume that the endpoints of
all intervals are unique. Otherwise we can suitably extend some of
the intervals on either side so that all endpoints become distinct
without altering the maximal cliques. We construct the sorted array
$A(1,\ldots,2n)$ of all endpoints in $O(n)$ time using bucket sort. The
maximal cliques are identified as follows. Traverse $A$ left to right and
whenever $A(i)$ is $Left(I_v)$ and $A(i+1)$ is $Right(I_u)$ for some $u,v$
in $V$ then output $[u,v]$. Clearly $u$ and $v$ are adjacent and hence
$[u,v]$ is a clique. Since $u$ is the leftmost possible and $v$ is the
rightmost possible for such a clique, $[u,v]$ is a maximal clique. The
traversal takes $O(n)$ time.
\end{proof}

\section{Equivalence of Coloring and Partition on Chordal Graphs}
\label{sec:eqv_chordal}

\begin{lemma}
 If a graph $G$ has a $[\lambda, C]$-coloring then it has a
 $[\lambda,C]$-partition.  \label{thm:coloring_to_partition}
\end{lemma}

\begin{proof}
 Suppose $G$ has a $[\lambda, C]$-coloring $\mathcal{C}$.  Consider
 the partition $\Pi$ induced by $\mathcal{C}$ where each part is
 exactly a chromon. The connectedness constraint is immediately
 satisfied. Since a chromon has size at most $C$, the size constraint
 is also satisfied. Since any pair of vertices in a clique is directly
 connected by an edge, the chromons in $\mathcal{C}$ intersected by a
 clique are all of different colors. Hence, a clique intersects at most
 $\lambda$ parts in $\Pi$. Thus the clique intersection constraint is
 also satisfied. Hence, $\Pi$ is a $[\lambda,C]$-partition.
\end{proof}

Next we will show that for chordal graphs the converse is
also true.

\begin{lemma}
 If a chordal graph $G$ has a $[\lambda, C]$-partition then it has a
 $[\lambda,C]$-coloring.  \label{lem:coloring_eqv_partition}
\end{lemma}

\begin{proof}
 Suppose $G$ has a $[\lambda, C]$-partition $\Pi = \{P_1, P_2, \ldots,
 P_t\}$. We prove that there exists a $[\lambda,C]$-coloring of $G$
 in which each $P_i$ is a chromon. Let the colors be numbered $1, 2,
 \ldots$. We prove by induction on number of vertices $n$. For $n=1$,
 assigning color $1$ to the single vertex gives a $[\lambda, C]$-coloring
 for any $\lambda, C \ge 1$.
 
 For $n>1$, let $u$ be a simplicial vertex of $G$. Without loss of
 generality assume $u \in P_1$. Consider the graph $G'$ obtained by
 removing $u$ from $G$. Then $\Pi'=\{P_1 \setminus \{u\}, P_2, \ldots,
 P_t\}$ is a $[\lambda,C]$-partition for $G'$. By induction, there is
 a $[\lambda,C]$-coloring $\mathcal{C}'$ of $G'$ in which each part
 of $\Pi'$ is a chromon. We obtain a coloring $\mathcal{C}$ of $G$ as
 follows. If $|P_1|>1$ we assign the color of other vertices in $P_1$
 to $u$ too. Otherwise we assign $u$ the lowest numbered color that is
 not assigned to any of the neighbors of $u$ in $\mathcal{C}'$. To show
 that $\mathcal{C}$ is a $[\lambda,C]$-coloring, it is enough to show
 that at most $\lambda$ colors are used in $\mathcal{C}$. For $|P_1|>1$
 it is obvious as no new color is used. For $|P_1|=1$ if it requires
 $\lambda+1$ colors then it implies that the clique $u \cup N(u)$
 intersects $\lambda+1$ parts which is not possible.
\end{proof}

Thus, on chordal graphs, solving the coloring problem is equivalent
to solving the partition problem. In the rest of the paper we solve
the partition problem only because the solution can be converted to a
solution to the coloring problem using the procedure described in the
proof of Lemma~\ref{lem:coloring_eqv_partition}.


\section{Equivalence of Coloring and Block-partition on PIGs}
\label{sec:eqv_pig}

For PIGs, we introduce a more restricted way of partitioning the vertex
set. 
\begin{definition}
A PIG is said to have a \emph{$[\lambda,C]$-block partition} if
it has a $[\lambda,C]$-partition in which each part also satisfy
\emph{consecutiveness} constraint, \ie each part is also a block.
\end{definition}

\begin{lemma}
A PIG $G$ has a $[\lambda, C]$-partition if and only if $G$ has a
$[\lambda, C]$-block partition.
 
\label{lem:partition_eqv_block} 
\end{lemma}

\begin{proof}
A $[\lambda, C]$-block partition is also a $[\lambda, C]$-partition. Now
suppose $G$ has a $[\lambda, C]$-partition $\Pi$.  If the parts in $\Pi$
also satisfy the consecutiveness constraint, we are done.  So assume
not.  We convert $\Pi$ to a new partition $\Pi'$ that also satisfies
the consecutiveness constraint.  The conversion is done by exchanging
vertices among the parts in $\Pi$, step-by-step, as follows.

We call a vertex $u$ to be \emph{terminal} if $u$ and some $v >
u+1$ belong to one part but $u+1$ belongs to a different part,
\emph{non-terminal} otherwise.  Let $P_1$ be the leftmost part whose
vertices are not consecutive.  Let $i\in P_1$ be smallest terminal vertex
such that $i+1$ is in some $P_2\ne P_1$, and there exists $i+k\in P_1$
for some $k>1$.  We will show how to repartition $P=P_1\cup P_2$ into
parts $P_1'$ and $P_2'$ such that in the new partition, each vertex
in the range $[1,i]$ is a non-terminal vertex. Then by repeating this
process all vertices can be made non-terminal and hence consecutiveness
constraint will be satisfied.  Note that $P$ is connected as both $P_1,
P_2$ are connected and $P_2$ has a vertex in between two vertices of
$P_1$. There are two cases.

Case 1: There are at most $C$ vertices in $P$ to the right of $i$.
In this case we set $P_2'$ to be the vertices in $P$ to the right
of $i$, and the $P_1'$ to be the vertices in $P$ to the left of and
including $i$.  Clearly, $i$ is no more a terminal vertex.  Since $P$
is connected, the vertices of $P$ considered in the canonical ordering
form a path. $P_1',P_2'$ are formed by breaking this path in the middle,
so $P_1',P_2'$ are both connected.  Let $Q$ be any maximal clique
which intersects $P_1',P_2'$.  Since we know that the vertices of $Q$
are consecutive, and $i$ is the rightmost vertex in $P_1'$ and $i+1$
the leftmost vertex in $P_2'$, the vertices $i,i+1$ must be in $Q$.
Thus $Q$ intersects $P_1,P_2$ as well.  All other parts intersecting $Q$
remain unchanged, so the number of parts intersected by $Q$ is the same
in the new partition as the old.

Case 2: There are more than $C$ vertices in $P$ to the right of $i$.
In this case we set $P_2'$ to be the $C$ rightmost vertices in $P$,
and the remaining go to $P_1'$.  As before we see that~$i$ is no more
a terminal vertex and $P_1',P_2'$ satisfy the connectedness property.
Consider a maximal clique $Q$ that intersects $P_1',P_2'$.  We show
that it must intersect the same number of parts in the new partition
as the old.  $Q$ must contain the rightmost vertex $u$ of $P_1'$ and
leftmost vertex $v$ of $P_2'$.

Note first that $P_1'$ contains both $i,i+1$, \ie it has at least
one vertex from $P_1$ and one vertex from $P_2$.  But $P_2'$ has $C$
vertices, so they cannot all be from $P_1$, or all from $P_2$ because
both $P_1,P_2$ had at most $C$ vertices each.  Thus $P_2'$ also contains
at least one vertex $j$ from $P_1$ and one vertex $k$ from $P_2$.
Since $i,j\in P_1$, there must be a path in $P_1$ from $i$ to $j$.
There must exist an edge $(u',v')$ in this path such that $u'\le u$,
and $v\le v'$ (see Fig.~\ref{fig:cli}).  Since $Q$ is maximal, it must
contain $u',v'$.  Thus $Q$ intersects $P_1$.  In a similar manner, we
see that it must intersect $P_2$.  Thus it follows that $Q$ intersects
the same number of parts in the old and new partitions.
\end{proof}

\tikzstyle{vertex}=[draw, thick, circle, minimum size=4mm]
\tikzstyle{maxcliq}=[draw, thick, ellipse, inner sep=0pt]

\begin{figure}[h!tb]
 \centering
 \begin{tikzpicture}
  \node[draw, circle, minimum size=4mm] (i) at (1cm,0cm) {};
  \node[draw, circle, minimum size=4mm] (ii) at (2cm,0cm) {};
  \node[draw, circle, minimum size=4mm] (u) at (4cm,0cm) {};
  \node[draw, circle, minimum size=4mm] (v) at (6cm,0cm) {};
  \node[draw, circle, minimum size=4mm] (vp) at (9cm,0cm) {};
  \node[draw, circle, minimum size=4mm] (k) at (8cm,0cm) {};
  \node[draw, circle, minimum size=4mm] (j) at (11cm,0cm) {};
  \node[above=2mm] at (i) {$i=u'$};
  \node[above=2mm] at (ii) {$i+1$};
  \node[above=2mm] at (u) {$u$};
  \node[above=2mm] at (v) {$v$};
  \node[above=2mm] at (vp) {$v'$};
  \node[above=2mm] at (k) {$k$};
  \node[above=2mm] at (j) {$j$};

  \draw[thick, decorate,decoration={brace,raise=8mm,amplitude=6pt}] (i.center) -- node[above=10mm] {$P_1'$} ++(3,0);
  \draw[thick, decorate,decoration={brace, raise=8mm,amplitude=6pt}] (v.center) -- node[above=10mm] {$P_2'$} ++(5,0);

  \draw[thick] (i.south) .. controls +(down:8mm) and +(down:8mm) .. (vp.south);
  \draw[thick] (vp.south) .. controls +(down:4mm) and +(down:4mm) .. (j.south);
  \draw[thick] (ii.south) .. controls +(down:5mm) and +(down:5mm) .. (k.south);
 \end{tikzpicture}
 \caption{Sketch showing clique intersection remains unchanged}
 \label{fig:cli}
\end{figure}
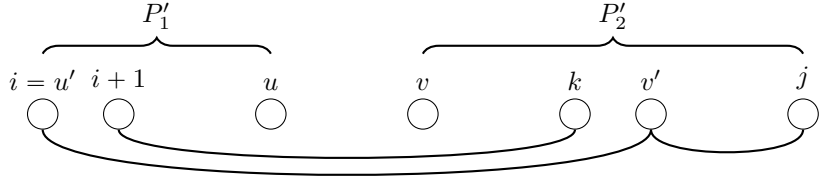

There is a simple example of a general (non-proper) interval graph
where Lemma~\ref{lem:partition_eqv_block} does not work. Consider the
example graph given by the intervals in canonical ordering: $a=[1,9],
b=[2,5], c=[3,6], d=[4,12], e=[7,10], f=[8,11]$. It has two maximal
cliques $Q_1=\{a,b,c,d\}, Q_2=\{a, d, e, f\}$. For $C=2$, the optimal
partition $\{\{a,d\}, \{b,c\}, \{e,f\}\}$ has clique intersection $2$ but
the part $\{a,d\}$ is not a block as $a,d$ are not consecutive according
to canonical ordering. All block partitions have clique intersection $3$
or more. The reason is as follows. If $a$ is the only vertex in a part
then to cover the remaining $3$ vertices of $Q_1$ we need at least 2
more parts. On the other hand if $a$ is paired with $b$ then to cover
the remaining $3$ vertices of $Q_2$ we need at least 2 more parts.

Since for a PIG, the notions of partition and block partition are
equivalent, in the rest of the paper we will abuse the notation $[\lambda,
C]$-partition to actually mean a $[\lambda, C]$-block partition in the
context of PIGs.

\begin{lemma}
Given a PIG $G$ with an interval representation, if there is an $O(f(n))$
algorithm to solve the (block) partition problem on $G$, then there
is an $O(n+f(n))$ algorithm to solve the coloring problem on $G$.
\label{lem:color_using_part}
\end{lemma}

\begin{proof}
We first get the canonical ordering of the vertices using the procedure
given in Proposition~\ref{pro:algocanon}. Suppose the partition algorithm
returns a partition with clique intersection $\lambda$ and the parts
sorted in canonical ordering are $P_1, P_2, \ldots, P_t$. For each $1 \le
i \le t$, we assign color $(i-1) \bmod \lambda +1$ to $P_i$. This is a
valid coloring because otherwise, there is an edge $(u,v)$ between two
parts of same color implying the clique $[u,v]$ in $G$ intersects more
that $\lambda$ parts which is not possible in a $[\lambda,C]$-partition.
\end{proof}


\section{An LP Based Algorithm for Block-partition on PIGs}
\label{sec:lp}

Let $G=(V,E)$ be a PIG with vertices in $V$ already sorted in canonical
ordering and $\mathcal{Q}$ be the set of maximal cliques. Let $Left(Q)$
denote the leftmost vertex of $Q$. Then partition problem on $G$ can be
formulated as the following integer linear program:
\begin{align}
 \text{\pilp:} \quad \text{min} \quad \lambda  & \nonumber \\
 \text{s.t.} \quad
 x_n & = 1 \label{eq:x1}\\
 \sum_{j=i}^{i+C-1} x_j &\ge 1 && 1 \le i \le n-C+1
 \label{eq:size}\\
 \sum_{j=Left(Q)+1}^{Left(Q)+|Q|-1} x_j & \le
 \lambda-1 && \forall\; Q \in \mathcal{Q}  \label{eq:clique}\\
 x_j & \in \{0, 1\} && 1 \le j \le n \label{eq:01} \\
 \lambda & \quad \text{integer} \label{eq:int}
\end{align}
where $x_j$ is a binary variable to denote if vertex $j$ is the rightmost
vertex of a block and $\lambda$ denotes maximum clique intersection by
any clique. Constraint~\eqref{eq:x1} ensures that some block must end
at $n$. Constraints~\eqref{eq:size} ensure that among $C$ consecutive
vertices there must be at least one vertex which is the rightmost vertex
of a block because a block has size at most $C$. Since a clique $Q$
intersects at most $\lambda$ blocks, constraints~\eqref{eq:clique}
ensure that the vertices in $Q$, except the rightmost, can include the
rightmost vertices of at most $\lambda-1$ blocks. The objective is to
minimize the maximum clique intersection $\lambda$.

Let \plp be the LP relaxation of \pilp obtained by making $x_j$
a real variable in $[0,1]$ and making $\lambda$ unconstrained.

\begin{lemma}
 If $x,\lambda$ is a fractional solution to \plp then it can be rounded
 to a integer feasible solution $\bar{x},\bar{\lambda}$ in polynomial time.
 \label{lem:rounding}
\end{lemma}

\begin{proof}

Consider the following rounding scheme which takes $O(n)$ time. We use
a set of intermediate variables $y_0, y_1, \ldots, y_n$. We set
\begin{align*}
  y_0 = 0, \quad y_j = \sum_{i=1}^j x_i, \quad \bar{\lambda}=\floor{\lambda} \quad \text{ and } \quad \bar{x}_j = \left\{ \begin{array}{ll} 1 & \text{if} \; \ceil{y_{j-1}} \ne \ceil{y_j} \\ 0 & \text{otherwise} \end{array} \right. \text{for all} \; 1 \le j \le n.
\end{align*}

Note that each $\bar{x}_j$ is a 0-1 variable and $\bar{\lambda}$ is an
integer. Since $x_1=1$, by construction $\bar{x}_1=1$. Hence $\bar{x}$
satisfies constraint \eqref{eq:x1}.

Now we prove that $\bar{x}$ satisfies the constraints in
\eqref{eq:size}. Since $x$ satisfies $j$th of such constraints,
$x_j + x_{j+1} + \ldots + x_{j+C-1}  \ge 1$, \ie $y_{j+C-1}-y_{j-1}
\ge 1$. So there must be at least one index $k$ in $[j, j+C-1]$ such
that $\ceil{y_{k-1}} \ne \ceil{y_k}$ implying that $\bar{x}_k=1$. Thus
$\bar{x}$ also satisfies the $j$th constraint in \eqref{eq:size}.

Finally we prove that $\bar{x},\bar{\lambda}$ satisfy constraints in
\eqref{eq:clique} too. Consider the constraint for clique $Q$ and let
$j=Left(Q)$. Since $x$ satisfies this constraint, $x_{j+1} + x_{j+2}
+ \ldots + x_{j+|Q|-1} \le \lambda-1$, \ie $y_{j+|Q|-1} - y_j \le
\lambda-1$. So there can be at most $\lambda-1$ indices $k$ in $[j+1,
j+|Q|-1]$ such that $\ceil{y_{k-1}} \ne \ceil{y_k}$ implying that at
most $\floor{\lambda}-1$ of the corresponding $\bar{x}_{k}$s are set
to $1$. Thus $\bar{x}, \bar{\lambda}$ also satisfy the constraint for
clique $Q$.
\end{proof}

Clearly $\bar{\lambda} \le \lambda$. Since the integer
objective value $\bar{\lambda}$ cannot be strictly less than
the fractional value $\lambda$, $\bar{\lambda}=\lambda$ and hence
$\bar{x},\bar{\lambda}$ optimally solve \pilp. Thus by solving \plp
and rounding the solution using the procedure given in the proof of
Lemma~\ref{lem:rounding} gives a polynomial time algorithm for the
partition problem.


\section{A Combinatorial Algorithm for Block-partition on PIGs}
\label{sec:comb}

We now give a combinatorial algorithm for the partition problem on
PIGs. The algorithm does not use LP scaffolding and hence is more
efficient.

\subsection{Lower Bound}
\label{ssec:lb}

\begin{lemma}
If a PIG $G$ has a $[\lambda, C]$-partition then $\lambda \ge
\floor{(\omega(G)+C-1)/C}$.  \label{lem:lb}
\end{lemma}

\begin{proof}
Let $Q$ be a maximum clique of $G$, \ie $|Q|=\omega(G)$. To cover all
vertices of $Q$ by parts of size at most $C$, we need at least $\lceil
\omega(G)/C \rceil$ parts.  Hence, clique intersection $\lambda \ge
\ceil{\omega(G)/C}=\floor{(\omega(G)+C-1)/C}$.
\end{proof}

There is a simple example where $\lambda = \ceil{\omega/C}$ is not
enough to have a $[\lambda,C]$-partition. Consider the graph given by
the intervals $\{[1,3], [2,5], [4,6]\}$. Here $\omega=2$. For $C=2$,
the number of parts given by lower bound $\ceil{\omega/C}=1$ is not
enough as the single part would contain $3>C$ connected vertices.

\subsection{Upper Bound}
\label{ssec:ub}

\begin{lemma}
If an interval representation is given for a connected PIG $G$ then
there exists an algorithm that produces a $[\ceil{(\omega(G)+C-1)/C},
C]$-partition.  \label{lem:greedy_pig}
\end{lemma}

\begin{proof}
Consider the following algorithm which we call \simp. We first
arrange the vertices of $G$ in the canonical ordering time using
Proposition~\ref{pro:algocanon}. Then we assign the block of vertices
$[(i-1)C+1, \min\{n,iC\}]$ to part $P_i$ for each $1 \le i \le
\ceil{n/C}$.
  
Each part $P_i$ produced by \simp clearly has consecutive vertices and has
size at most $C$. Since $G$ is connected, by Corollary~\ref{cor:edge_seq},
there is an edge between $v_j$ and $v_{j+1}$ for all $j=1, 2, \ldots,
n-1$.  Hence $P_i$ is also connected. Thus it will be enough to show that
any clique intersects at most $\lambda=\ceil{(\omega+C-1)/C}$ parts.

If a clique $Q$ intersects $\lambda$ parts $P_{i}, \ldots,
P_{i+\lambda-1}$, then $Q$ must contain at least one vertex of
each of $P_{i}$ and $P_{i+\lambda-1}$ and all vertices of remaining
parts $P_2,\ldots,P_{i+\lambda-2}$. So the minimum size of $Q$ is
$1+(\lambda-2)C+1 = \lambda C -2C +2$. Thus $\omega \ge \lambda C -2C
+2$. Hence $\lambda \le (\omega+2C-2)/C$. This implies $\lambda \le
\floor{(\omega+2C-2)/C} = \ceil{(\omega+C-1)/C}$.
\end{proof}

There is a simple example where \simp does not give the optimal
partition. Consider the graph given by the intervals $a=[1,6], b=[2,7],
c=[3,10], d=[4,11], e=[5,12], f=[8,13], g=[9,14]$. It has two maximal
cliques $\{a,b,c,d,e\}, \{c,d,e,f,g\}$ and $\omega=5$. For $C=3$,
\simp produces the partition $\{\{a,b,c\}, \{d,e,f\}, \{g\}\}$ which has
clique intersection $3$. But there exists a better partition $\{\{a,b\},
\{c,d,e\}, \{f,g\}\}$ with clique intersection $2$.

However, a close analysis reveals that \simp is not that bad. In fact,
when $\omega(G)=kC+1$ for some integer $k$, the two bounds match and hence
\simp gives the optimal solution. Again for other values of $\omega(G)$,
which can be represented as $kC+r$ for integer $k,r$ such that $2\le r\le
C$, the two bounds are $k+1$ and $k+2$ respectively, hence differ by $1$
and one of the two bounds is optimum. Thus it will be enough to solve
the following special case of the problem.
\begin{definition}
\emph{Partition subproblem:} given a PIG $G$ with $\omega(G)=kC+r$, $k$
integer and $2\le r\le C$, check if there is a $[k+1,C]$-partition and
if so, generate the partition.
\end{definition}
If there is an algorithm \alg for the partition subproblem then
we apply \alg to $G$ to check if there is a $[k+1,C]$-partition. If
yes, \alg also gives the required partition. Otherwise, \simp gives an
optimal solution.

In the rest of the paper we will let $k(G)$, or in short $k$, denote
$\floor{(\omega(G)-1)/C}$.

\subsection{Forbidden vertices}
\label{ssec:direct}

The key idea in our algorithm is to first identify those vertices that
cannot be right endpoints of a block in a possible $[k+1,C]$-partition.

\begin{definition}
A vertex $i$ in a PIG is said to be \emph{primarily forbidden} if the
block $[i-kC,i+1]$ is a clique.
\end{definition}

\begin{lemma}
If the vertex $i$ in a PIG is primarily forbidden then no block in a $[k+1,C]$-partition can end at $i$.
\label{lem:forbid1}
\end{lemma}

\begin{proof}
Suppose a block in a $[k+1,C]$-partition ends at the vertex $i$. Consider
the clique $[i-kC,i+1]$ which must be covered by at most $k+1$ blocks. To
cover the vertex $i+1$ we need one block. Then the remaining $kC+1$
vertices $[i-kC,i]$ must be covered by at most $k$ blocks. This is not
possible as the size of a block is at most~$C$.
\end{proof}

\begin{definition}
A vertex in a PIG is \emph{forbidden} if it is primarily forbidden or
secondarily forbidden, where secondarily forbidden vertices are defined
recursively as follows.  If (a) there exists a block of forbidden vertices
$[v-s+1,v]$ where $1 \le s \le C-1$, and (b) the block $[v-kC,v-s+1]$
is a clique then the set of vertices $P(v)=\{v-qC \mid 1 \le q \le
k\}$ is \emph{secondarily forbidden}.  Further, we will say that $v$
is the \emph{leader} of all secondarily forbidden vertices in $P(v)$
and~$v$ itself. Similarly each secondarily forbidden vertex in $P(v)$
is a \emph{follower} of $v$.
\end{definition}

Note that a primarily forbidden vertex is the leader of
itself. Furthermore, any forbidden vertex $i$ has a leader $i+qC$ where
$q$ is an integer and $0 \le q \le k$.

\begin{lemma}
If the vertex $i$ in a PIG is secondarily forbidden then no block in a
$[k+1,C]$-partition can end at $i$.  \label{lem:forbid2}
\end{lemma}

\begin{proof}
The leader of $i$ is $i+qC$ where $1 \le q \le k$. Let $Q$ be the clique
$[i+qC-kC, i+qC-s+1]$. Since the vertices $[i+qC-s+1,i+qC]$ are forbidden,
they must be covered by a single block, say $B$. The block $B$ must end at
a vertex on the right of $i+qC$. In the best case $B$ ends at $i+qC+1$
and covers $C-s$ vertices of $Q$, \ie $[i+(q-1)C+2,i+qC-s+1]$. The
remaining $(k-1)C+2$ vertices $[i-(k-q)C,i+(q-1)C+1]$ of $Q$ must be
covered by at most $k$ blocks.

Now suppose a block of a $[k+1,C]$-partition ends at the vertex $i$. To
cover the $(k-q)C+1$ vertices $[i-(k-q)C,i]$ we need at least $(k-q)+1$
blocks. The remaining $(q-1)C+1$ vertices $[i+1,i+(q-1)C+1]$ must be
covered by at most $q-1$ blocks. This is not possible as the size of a
block is at most~$C$.
\end{proof}

Our algorithm is as follows.  We mark all forbidden vertices, and then
try to form blocks by a greedy left to right strategy.

\subsection{Marking forbidden vertices}
\label{ssec:direct}

The algorithm for marking forbidden vertices is given in
Algorithm~\ref{alg:mark}.  The algorithm assumes that we are given an
array $\lmn$, where $\lmn(i)$ denotes the leftmost neighbor of $i$.
It is easily seen that $\lmn$ can be computed in $O(n)$ time given an
interval representation of the input graph.  The algorithm constructs
the array $F$, where $F(i) = 1$ if and only if $i$ is forbidden.

\begin{algorithm}[h!tb]
 \SetKwInOut{input}{Input}
 \SetKwInOut{output}{Output}
 \SetKw{KwDownTo}{downto}
 \input{$\lmn(1,\ldots,n)$ for a PIG $G=(V,E)$}
 \output{$F(1,\ldots,n)$}
 \lForEach{$i=1~\emph{\KwTo}~n$}{ $F(i) = \Ldist(i) = 0$}\;
 \ForEach(\tcc*[f]{phase~1}){$i = n~\emph{\KwDownTo}~1$}{
  \If{$\lmn(i) \le i-kC-1$}{
   $F(i-1) = 1$\;
  }
 }
 $\Rnf(n) = n$\tcc*{$\Rnf(i)$ is rightmost non-forbidden vertex $j$ where $j \le i$}
 \ForEach(\tcc*[f]{phase~2}){$i = n-1~\emph{\KwDownTo}~1$}{
  $\Rnf(i) = \min\{i,\Rnf(i+1)\}$\;
  \lWhile(\tcc*[f]{Extend $\Rnf(i)$}){$F(\Rnf(i)){==}1$}{$\Rnf(i)=\Rnf(i)-1$}
  \If(\tcc*[f]{$i$ is a follower}){$(F(i){==}1)$ and $(\Ldist(i) \le (k-1)C)$}{
   $F(i-C) = 1;\quad$
   $\Ldist(i-C) = \Ldist(i)+C$\;
  }
  \If(\tcc*[f]{$i$ is a leader}){$\lmn(\Rnf(i)+1) \le i-kC$}{
   $F(i-C) = 1;\quad$
   $\Ldist(i-C) = C$\;
  }
 }
\caption{\mark}
\label{alg:mark}
\end{algorithm}

In phase~1, the primarily forbidden vertices are marked. We check at each
vertex $i$ if there is a clique $[i-kC-1,i]$ of size $kC+2$, that is, if
$\lmn(i) \le i-kC-1$, then we mark the primarily forbidden vertex $i-1$
by setting $F(i-1)=1$.

In phase~2, we mark the secondarily forbidden vertices. It is enough to
identify all the leaders because their followers are all the secondarily
forbidden vertices. We identify the leaders and mark their followers in
an interleaved manner in a single traversal through the vertices. If
$i$ is identified as a leader, we mark its rightmost follower $i-C$
immediately, we mark the second rightmost follower $i-2C$ when we visit
$i-C$, and so on.

Note that the rightmost vertex $n$ is never forbidden and hence is not
a leader. So in phase 2, we visit each vertex $i$ starting with the
second rightmost.  We check if the vertex is a follower of a already
discovered leader.  This is easy to do, we merely check if $i$ is a
follower of a vertex that is not far, \ie at most a distance $(k-1)C$
from $i$. Then we mark $i-C$. For this we maintain the auxiliary array
$\Ldist(1,\ldots,n)$ where $\Ldist(i)$ is the distance of $i$ from
its leader if $i$ is secondarily forbidden, unspecified otherwise.
Then we also check if $i$ is itself a leader.  The condition for $i$
being a leader is that there should exist a block of forbidden vertices
$[j,i]$, and a clique $[i-kC,j]$ where $i-j+1<C$.  However, as following
Lemma shows, we can assume without loss of generality that the block of
forbidden vertices is left maximal.

\begin{lemma}
In a PIG the vertex $i$ is a leader if and only if the left maximal
forbidden block at $i$ is $[j,i]$ and the block $[i-kC,j]$ is a clique.
\label{lem:forbid_q}
\end{lemma}

\begin{proof}
$[\Leftarrow]$ If the block $[i-kC,j]$ is a clique and $[j,i]$ is a
forbidden block then clearly $i$ is a leader.

$[\Rightarrow]$ If $i$ is a leader then there exists $j'$ such that $0 \le
i-j' < C-1$, the block $[j',i]$ is forbidden and the block $[i-kC,j']$
is a clique. If $[j',i]$ is not left maximal then let $[j,i]$ be the
left maximal forbidden block at $i$. Then $j<j'$. Hence, $[i-kC,j]$
is a subclique of $[i-kC,j']$.
\end{proof}

Now, checking if a vertex is a leader is easy.  We need to know the
leftmost endpoint $j$ of a block of forbidden vertices ending at $i$, and
whether a clique starting at $i-kC$ ends at $j$. The leftmost endpoint
of the forbidden block need not be calculated afresh for every $i$.
If the leftmost endpoint $j$ of the forbidden block ending at $i$
is already calculated, then we only need to check if the forbidden
block extends further on the left of $j$, when considering the vertex
$i$. For this we maintain the auxiliary array $\Rnf(1,\ldots,n)$ where
$\Rnf(i)$ is the rightmost non-forbidden vertex such that $\Rnf(i)\le
i$. The elements of $\Rnf$ can be recursively computed as follows:
$\Rnf(i)$ is the rightmost non-forbidden vertex $x$ such that the block
$[x+1,\min\{i,\Rnf(i+1)\}]$ is forbidden. To make sure the existence of
$\Rnf(i)$ for all $1 \le i \le n$, We assume without loss of generality
an imaginary vertex numbered $0$ that is not forbidden.

\begin{lemma}
If the array $\lmn$ for a PIG $G$ is given, the algorithm
\mark correctly marks the forbidden vertices of $G$ in $O(n)$ time.
\label{lem:mark}
\end{lemma}

\begin{proof}
It is easy to see that phase~1 of \mark correctly computes correct values
of $F$ for primarily forbidden vertices.  Let us refer each iteration of
the loop in phase~2 by the corresponding value of $i$. We now claim that,
at the beginning of iteration~$i$ in phase~2, the array $F$ contains
correct values for all vertices in the range $[\max\{1,i-C+1\},n]$. This
in turn proves the correctness of \mark. We show by induction
on~$i$.

For $i=n$ there cannot be any secondarily forbidden vertex in the range
$[i-C+1,n]$. Hence the claim is trivially true. For $i \le n$ assume at
the beginning of iteration~$i$, the array $F$ contains correct values in
the range $[i-C+1,n]$ (for simplicity we assume $i \ge C+1$, the cases
$i\le C$ can be shown similarly). The vertex $i-C$ is a secondarily
forbidden vertex if and only if it is either the right most follower of
$i$ or it is a follower of some vertex on the right of $i$. In \mark we
handle the second case first and update $F$ accordingly. After adjusting
$\Rnf(i)$ suitably, $[\Rnf(i)+1,i]$ correctly denotes the left maximal
forbidden block at $i$ because by induction hypothesis the forbidden
vertices in the range $[i-C+1,i]$ are already marked. If there is
a clique $[i-kC,j+1]$ then $i$ is a leader and we update $F$ for the
rightmost follower $i-C$ of $i$. Hence \mark correctly computes if $i-C$
is a forbidden vertex in iteration~$i$. Thus at the end of iteration $i$,
\ie at the beginning of iteration $i-1$, the array $F$ contains correct
values in the range $[(i-1)-C+1,n]$.

The pseudocode of Algorithm~\ref{alg:mark} clearly shows that \mark
takes overall $O(n)$ time.
\end{proof}

\subsection{Algorithm \combp}
\label{ssec:direct}

We now give our algorithm to solve the partition subproblem. The algorithm
first marks all forbidden vertices and then forms blocks greedily such
that no block ends at a forbidden vertex. We call this algorithm \combp,
which is shown in Algorithm~\ref{alg:combp}.

\begin{algorithm}[h!tb]
 \SetKwInOut{input}{Input}
 \SetKwInOut{output}{Output}

 \input{A PIG $G=(V,E)$}
 \output{If $G$ has a $[k+1, C]$-partition; if \yes also output such a partition}

 $F(1,\ldots,n)=$ array returned by \mark on $G$; $\quad u=1$\;
 \While{$u\le n$}{
 $v=\min\{u+C-1, n\}$\;
   \lWhile{$(v \ge u)$ and $(F(v) == 1)$}{$v = v-1$}\;
   \lIf{$v<u$}{\Return{\no}} \lElse{create part $[u,v]$; $u=v+1$}\;
 }
 \Return{\yes}\;
\caption{\combp}
\label{alg:combp}
\end{algorithm}

\begin{lemma}
Let $B$ be any block, except the rightmost, created by \combp. Let
$u$ be the leftmost vertex of $B$. Then the block of vertices
$[u+jC+|B|,u+jC+(C-1)]$ are forbidden for $0 \le j \le k-1$.
\label{lem:forbid4}
\end{lemma}

\begin{proof}
If $|B|=C$ then the block $[u+jC+|B|,u+jC+(C-1)]$ is empty and
hence the lemma is vacuously true. So we assume $|B| < C$. Note that
$u+|B|,\ldots,u+C-1$ are all forbidden because otherwise \combp would
have created the block $B$ of bigger size. Also note that either $B$
is the leftmost block or $u-1$ is the rightmost vertex of a block. So
without loss of generality we assume that $u-1$ is not forbidden.

It will be enough if we prove that for each $|B| \le t \le C-1$ the
vertex $u+kC+t$ is a leader because then its followers $P(u+kC+t)$,
\ie $u+t, u+C+t, \ldots, u+(k-1)C+t$ are all forbidden. We prove by
induction on $t$.

Base case: $t=C-1$. Since $u+C-1$ is forbidden, its leader is the
vertex $v=u+C-1+qC$ for some $0 \le q \le k$ and the followers of $v$,
the vertices in $P(v)$, are forbidden. But $u-1$ is not forbidden, \ie
$u-1 \notin P(v)$. Hence $u-1 < v-kC$, implying $q>k-1$. Thus $q=k$,
which implies our claim.

Induction case: suppose the claim is true for $t=t'$ where $|B| < t'
\le C-1$, \ie the vertex $u+kC+t'$ is a leader. We need to prove that
$u+kC+t'-1$ is also a leader.

Since $u+t'-1$ is forbidden, its leader is the vertex
$v=u+t'-1+qC$ for some $0\le q\le k$. Thus, if $q=k$
then we are done.  So assume that $q<k$, \ie $q=k-z, 1 \le z \le k$.
We will show that this leads to a contradiction.

Since $v$ is a leader, by Lemma~\ref{lem:forbid_q}, for some vertex $x$,
the block $F_1=[x,v]$ is the left maximal forbidden block at $v$ and
the block $Q=[v-kC,x]$ is a clique. Again, by induction hypothesis,
each of the vertices $[u+kC+t',u+(k+1)C-1]$ is a leader. Thus,
the set of vertices $F_2=[u+qC+t',u+(q+1)C-1]$ is forbidden. But
$F_2$ can be rewritten as $[v+1,u+(q+1)C-1]$. Thus $F= F_1\cup F_2 =
[x,u+(q+1)C-1]$ is a left maximal forbidden block at $u+(q+1)C-1$. By
applying Lemma~\ref{lem:forbid_q} to $Q$ and $F$, the vertex $u+(q+1)C-1$
is a leader. Among its followers, $P(u+(q+1)C-1)$, the $z$th
from the left is $u+(q-k+z)C-1=u-1$. This is a contradiction because $u-1$
is not forbidden.
\end{proof}

\begin{lemma}
 Suppose \combp creates consecutive blocks $B_1, \ldots, B_{k+1}$. Let $u$
 be the leftmost vertex of $B_1$. Let $s_i=C-|B_i|$ and $S_i=\sum_{j=1}^i
 s_j$. Then the $S_k$ consecutive vertices $[u+kC-S_k,u+kC-1]$
 are all forbidden.  \label{lem:forbid3}
\end{lemma}

\begin{proof}
For $1 \le i \le k$, let $u_i$ be the leftmost vertex of block $B_i$. Note
$u_i=u+(i-1)C-S_{i-1}$  where $S_0=0$.

Applying Lemma~\ref{lem:forbid4} to $B_i$ for all $1 \le i \le k$ and
considering the set of consecutive forbidden vertices $F_i$ corresponding
to $j=k-i+1$ we get $F_i = [u_{i}+(k-i+1)C-s_i,u_{i}+(k-i+1)C-1] =
[u+kC-S_{i},u+kC+S_{i-1}-1]$.

The set $\cup_{i=k}^1 F_i$ is indeed the required set of consecutive
forbidden vertices.
\end{proof}

\begin{lemma}
 If an interval representation for a PIG $G$ is given then \combp
 correctly solves the partition subproblem on $G$ in $O(n)$ time.
 \label{lem:combp}
\end{lemma}

\begin{proof}
If \combp outputs NO, then there is a set of $C$ consecutive forbidden
vertices. To cover these vertices we need a block of size at least
$C+1$. So there cannot be any valid partition. Hence \combp is correct.

Now we prove that if \combp outputs \yes then the partition generated
is a valid partition.  Since the algorithm generates blocks of size at
most $C$, the size constraint is satisfied. We only need to prove that
no clique intersects more than $k+1$ blocks generated by \combp. We
prove this by contradiction.

Suppose there is a clique $Q$ that intersects $k+2$ blocks $B_0, B_1,
\ldots, B_{k+1}$. Without loss of generality, we assume that only
the leftmost vertex of $Q$ is covered by $B_0$ and only the rightmost
vertex of $Q$ is covered by $B_{k+1}$. Because, otherwise we can take a
sub-clique $Q' \subset Q$ with this property. Also let $v$ be the leftmost
vertex of $Q$, \ie rightmost vertex of $B_0$ and hence not forbidden. Let
$s_i=C-|B_i|$ and $S_i=\sum_{j=1}^i s_j$ for $1 \le i\le k$. Let $u$ be the
leftmost vertex of $B_1$.

Note $|Q| \le kC+1$ because otherwise $v$ would be forbidden. By
Lemma~\ref{lem:forbid3} there are $S_k$ consecutive forbidden
vertices. Since the algorithm outputs \yes, $S_k < C$. Hence $|Q| = 2 +
\sum_{i=1}^k |B_i| = 2 + \sum_{i=1}^k (C-s_i) = 2+kC-S_k > (k-1)C+2$.

So $Q$ has size $kC+2-S_k$ where $1 \le S_k \le C-1$ and by
Lemma~\ref{lem:forbid3}, there are $S_k$ consecutive forbidden vertices
$[u+kC-S_k,u+kC-1]=[v+kC+1-S_k,v+kC]$. By Lemma~\ref{lem:forbid2}, $v$
is forbidden.  It is a contradiction.

If an interval representation is given then we can easily find the
maximal cliques and hence can compute the array $\lmn$ in $O(n)$
time. By Lemma~\ref{lem:mark} the marking of forbidden vertices takes
time $O(n)$. The greedy procedure for generating the parts also takes
$O(n)$ time. Overall time taken is $O(n)$.
\end{proof}

Combining Lemma~\ref{lem:color_using_part}, Lemma~\ref{lem:combp} and
the discussions at the end of the subsection~\ref{ssec:ub} we get a
proof of Theorem~\ref{thm:pig}.


\section{Algorithm for Splittable Weighted Problem on PIGs}
\label{sec:splittable}

\begin{lemma}
A weighted graph $G=(V,E,W)$ is $[\lambda,C]$-split colorable if and only
if $WXP(G)$ is $[\lambda,C]$-colorable.  \label{lem:wtexpand}
\end{lemma}

\begin{proof}
Let $G' = (V',E') = WXP(G)$.

$[\Leftarrow]$ The weighted graph $G''$, obtained by putting weight $1$
to every vertex of $G'$, is also a weight-split graph of $G$. So if $G'$
is $[\lambda,C]$-colorable then $G''$ and $G$ both are $[\lambda,C]$-split
colorable.

$[\Rightarrow]$ Let the weighted graph $G''(V'',E'',W'')$ be the
weight-split graph corresponding to the $[\lambda,C]$-split coloring
of $G$. Then $G'$ is the weight-expanded graph of $G''$ too. A
$[\lambda,C]$-coloring of $G'$ can be obtained by assigning the vertices
in $G'$ corresponding to a vertex $v''$ in $G''$ the same color of $v''$.
\end{proof}

Thus solving the split coloring problem on a weighted PIG $G=(V,E,W)$ is
equivalent to solving the unweighted coloring problem on $WXP(G)$. In
the rest of the section we will use $G'=(V',E')$ to represent
$WXP(G)$. Applying the algorithm described in Section~\ref{sec:comb} on
$G'$ gives correct result but it makes the algorithm pseudo-polynomial
as it takes $O(n')$ time, proportional to the sum of weights. This is
mainly because the algorithm iterates over each vertex in $G'$.

However, it turns out that iterating over each vertex in $G'$ is not
necessary. The forbidden vertices in $G'$ can be divided into blocks
such that if the vertices $u$ and $v$ are in the same block $b$ then
leader of $u$ and leader of $v$ are in the same block $l$. We call such
forbidden blocks \emph{FB}s. Parallel to the vertices, we say that FB $l$
is the \emph{leader} of FB $b$ and $b$ is the \emph{follower} of $l$. It
can be seen that all vertices in an FB can be marked together. Hence it
is enough to iterate through the FBs instead of iterating through the
vertices of $G'$.

\subsection{Marking forbidden blocks}

We now modify the algorithm presented in Section~\ref{sec:comb} to
let it work with FBs instead of forbidden vertices.  The modified
algorithm to mark all the FBs, which we call \splitm, is shown in
Algorithm~\ref{alg:splitm}.

\begin{algorithm}[h!tb]
 \SetKwInOut{input}{Input}
 \SetKwInOut{output}{Output}
 \SetKw{KwDownTo}{downto}

 \input{Maximal cliques of a PIG $G$, $\lmn(1,\ldots,n)$, $Z(0,\ldots,n)$}
 \output{Doubly linked list of FBs $F$ in $G'=WXP(G)$}

 \ForEach(\tcc*[f]{phase~1}){maximal clique $[u,v]$ in $G$}{
 \If(\tcc*[f]{$[Z(u{-}1){+}1,Z(v)] \in G' \equiv [u,v] \in G$}){$Z(v){-}Z(u{-}1){-}(kC{+}2) \ge 0$}{
    $F.\inlay(Z(v)-1, Z(v)-Z(u-1)-(kC+2)+1, 0)$\;
  }
 }

 $i = F.\nd\ptr\prv$\tcc*{$F.\nd\ptr\prv$ is the rightmost FB}
 \While(\tcc*[f]{phase~2}){$i \ne F.\bg$}{
 $v = i\ptr\rit$\;
 $j = i\ptr\rnf = i\ptr\nxt\ptr\rnf$; \quad \lIf{$v<i\ptr\rnf\ptr\rit$}{$j = i\ptr\rnf = i$}\;
  \lWhile{$j\ptr\prv\ptr\rit == j\ptr\rit-j\ptr\size$}{ $i\ptr\rnf = j$; $j = j\ptr\prv$}\;
  \If(\tcc*[f]{$i$ is a follower block}){$i\ptr\ldist \le (k-1)C$}{
   $F.\inlay(v-C, i\ptr\size, i\ptr\ldist+C)$\;
  }
  $u=i\ptr\rnf\ptr\rit - i\ptr\rnf\ptr\size + 1$; \quad $f = v - u + 1$\;
  $s = (v{-}\lmn(\bar{h}(u)){+}1){-}(kC{+}2)$\tcc*{function $\bar{h}$ is computed using $Z$}
  \If(\tcc*[f]{new leader block ending at $v$}){$f+s \ge 0$}{
   $F.\inlay(v-C, f+s+1, C)$\;
  }
  $i = i\ptr\prv$\;
 }
\caption{\splitm}
\label{alg:splitm}
\end{algorithm}

We use the following correspondence between a vertex $v \in V$ and a
vertex $v' \in V'$. The vertex $v'=h(v,q)$ if $v'$ is the $q$th copy
of $v$ where $1 \le q \le W(v)$ and $v=\bar{h}(v')$ if $v'$ is a copy
of $v$. The set $\{h(v,1), \ldots, h(v,W(v))\}$ of copies of $v$ is
represented by $H(v)$. As usual, we will interchangeably use $1 \le v'
\le n'$ ($1 \le v \le n$) to denote a vertex $v' \in V'$ ($v\in V$)
as well as its position in the canonical ordering of vertices in $G'$
($G$). We also use an auxiliary array $Z(0,\ldots,n)$ such that $Z(0)=0$
and for all $v>0$, the entry $Z(v)$ denotes the rightmost copy of $v$ in
$G'$, \ie $h(v,W(v))$. Since all the copies $h(v,t) \in V'$ of $v\in V$
appear consecutively in the canonical ordering of $G'$, we have $Z(v)=
\sum_{i=1}^{v} W(i)$. Note that given $Z$, the values of the function
$\bar{h}(u)$ for all $u$ belonging to a subset of vertices $S \subseteq
V'$, can be computed in right to left order, in overall $O(|S|+n)$ time.

We store the information about the FBs in a linked list. Thus $F$ is a now
a doubly linked list of non-intersecting FBs sorted according to canonical
ordering. We also keep the information stored in auxiliary arrays $\Ldist$
and $\Rnf$ earlier, in the list $F$ itself. Thus each entry $b$ of $F$ has
the following fields: (i) $\rit$ denotes the rightmost vertex of the FB
$b$, (ii) $\size$ denotes the size of $b$, and (iii) $\ldist$ denotes the
distance of $b.\rit$ from its leader, (iv) $\rnf$ points to the leftmost
FB such that all FBs between $b.\rnf$ and $b$ are consecutive, \ie all
vertices in $[b.\rnf\ptr\rit-b.\rnf\ptr\size+1,b.\rit]$ are forbidden,
(v) $\prv$ points to the FB on the left of $b$, and (vi) $\nxt$ points
to the FB on the right of $b$. Note that we use the notation $p{\ptr}q$
to represent the field $q$ of the FB pointed by the pointer $p$. We keep
two sentinel FBs in $F$ always, (i) the leftmost FB $[-2,-1]$ and (ii)
the rightmost FB $[Z(n)+2, Z(n)+3]$ each having $\ldist=0$ and $\rnf$
pointing to itself. Two pointers $F.\bg$ and $F.\nd$ point to these two
FBs, respectively.

In addition to the standard operations of insert, delete and both way
traversals though the list, we define a new operation on $F$ which we
call $F.\inlay(rt,sz,ld)$. This operation inserts a new FB $b$ with
$b.\rit=rt, b.\size=sz, b.\ldist=ld$ into $F$ but makes sure that the
FBs in $F$ remain non-intersecting and sorted. Let $b_1, \ldots, b_s$
be the FBs in $F$ which intersect $b$. The operation $\inlay$ does the
following: (i) deletes all the FBs in $F$ which are subsets of $b$,
(ii) if $b$ partly intersects $b_1$, \ie if $t=(b.\rit-b.\size - b_1.\rit)
< b_1.\size$ then updates $b_1.\size=b_1.\size-t$, $b_1.\rit=b_1.\rit-t$,
(iii) if $b$ partly intersects $b_s$, \ie if $b_s.\rit>b.\rit$ then updates
$b_s.\size = b_s.\rit - b.\rit$, and (iv) inserts $b$ at its proper
position in $F$.

We slightly modify the definition of the array $\lmn(1,\ldots,n)$. Now
$\lmn(v)$ denotes the leftmost neighbor of $Z(v)$ in $G'$. If the maximal
cliques of $G$ are given, then the elements of $\lmn$ can be computed
in $O(n)$ time.

In phase~1 we mark the FBs due to the primarily forbidden vertices given
by the following lemma:

\begin{lemma}
Let a PIG have a maximal clique $Q$ with rightmost vertex $v$ and let
$s=|Q|-(kC+2) \ge 0$. Then the set of all primarily forbidden vertices
in $Q$ is exactly the FB $[v-s-1,v-1]$.
\end{lemma}

\begin{proof}
Let $u$ be the leftmost vertex of $Q$. The only subcliques of $Q$ which
satisfy the conditions of Lemma~\ref{lem:forbid1} are $Q_p = [u+p,
u+kC+1+p]$ for all $0 \le p \le s$. The primarily forbidden vertex for
$Q_p$ is $u+kC+p$. Thus the set of all phase~1 forbidden vertices for $Q$
is $\{u+kC+p, 0 \le p \le s\}$, \ie the FB $[v-s-1,v-1]$.
\end{proof}

So in phase~1 we go through the maximal cliques of $G'$ which have
one-to-one correspondence with the maximal cliques in $G$ and mark the
FBs for the primarily forbidden vertices. Note that the maximal clique
$[u,v] \in G$ corresponds to the clique $[Z(u-1)+1,Z(v)] \in G'$.

It is clear that in phase~2 we need not check for leaders at the vertices
which are not forbidden; checking only at the rightmost vertex in each
FB suffices.  The following lemma determines the size of the leader
block ending at the rightmost vertex in a FB.

\begin{lemma}
Let a PIG have a left maximal block $B$ of forbidden vertices with the rightmost
vertex $v$ and size $f$ where $1 \le f <C$. Let $Q$ be the largest clique
with the rightmost vertex $v-f+1$. Let $s=|Q|-(kC+2)$. Then the leaders in
$B$ are exactly the FB $[v-f-s, v]$.  \label{lem:forbid_q4}
\end{lemma}

\begin{proof}
Let $u$ be the leftmost vertex of $Q$. Thus $u=v-f+1-|Q|+1 =
v-kC-(f+s)$. Consider any vertex $v' \in B$. If $f+s<0$ then $u>v-kC
\ge v'-kC$. Since $Q$ is the largest possible there cannot be a clique
$[v'-kC,v-f+1]$. Hence by Lemma~\ref{lem:forbid_q}, the vertex $v'$
is not a leader. So we assume $f+s \ge 0$.

Note that $s<0$ because otherwise $v-f$ would be a primarily forbidden
vertex on the immediate left of $B$ which is not possible as $B$
is left maximal. Thus $[v-f-s,v]$ is a subset of $B$. Now for all
$0 \le p \le f+s$, consider the vertex $v_p=v-p$. Since $v_p-kC =
v-kC-p \ge u$, the block $[v_p-kC,v-f+1]$ is a subclique of $Q$. Thus
by Lemma~\ref{lem:forbid_q}, the vertex $v_p$ is a leader. Since $Q$ is
the largest possible, $[v-f-s, v]$ is exactly the set of leaders in $B$.
\end{proof}

Thus in phase~2 we visit each FB $i$ starting with the rightmost FB
created in phase~1. We check if $i$ is a follower of some previously
discovered FB $l$ at a distance at most $(k-1)C$ and mark the next
follower of $l$ on the left of $i$. We also check if $i$ is a leader
itself using Lemma~\ref{lem:forbid_q4}. If a leader FB $l$ is identified
then we insert the rightmost follower of $l$, and so on, similar to
secondarily forbidden vertices in Section~\ref{sec:comb}.

\begin{lemma}
If the arrays $Z, \lmn$ and the maximal cliques of a PIG $G$ are given,
then \splitm correctly marks the FBs of $G'=WXP(G)$ in $\Theta(n^2)$ time.
\label{lem:split_mark}
\end{lemma}

\begin{proof}
It is easy to see that phase~1 of \splitm correctly inserts into $F$
the FBs due to the primarily forbidden vertices while ensuring that the
FBs in $F$ are non-intersecting and sorted in canonical ordering. Let
us refer each iteration of the loop in phase~2 by the corresponding
value of $i$. We now claim that, at the beginning of iteration~$i$
in phase~2, the list $F$ correctly contains all FBs in the range
$[\max\{1,i\ptr\rit-C+1\},n']$. This in turn proves the
correctness of \splitm. We show by induction on~$i$.

For $i$ pointing to the rightmost FB there cannot be any secondarily
forbidden vertex in the range $[i\ptr\rit-C+1,n]$. Hence the claim
is trivially true. For other values of~$i$, assume at the beginning
of iteration~$i$, the list $F$ correctly contains FBs in the range
$[i\ptr\rit-C+1,n]$ (for simplicity we assume $i\ptr\rit \ge C+1$,
the cases $i\ptr\rit \le C$ can be shown similarly). In iteration
$i$ \splitm correctly inserts a new FB $[i\ptr\rit-i\ptr\size-C+1,
i\ptr\rit-C]$ or $[i\ptr\rit-(f+s+1)-C+1, i\ptr\rit-C]$ depending upon whether (i) FB $i$ is a follower of a
previously discovered FB or (ii) there is a leader FB with rightmost
vertex $i\ptr\rit$, given by Lemma~\ref{lem:forbid_q4}. In case~(i) there
are non forbidden vertex in $[b\ptr\rit+1, i\ptr\rit-i\ptr\size]$. Hence
there can not be any secondarily forbidden vertex in the $[b\ptr\rit-C+1,
i\ptr\rit-i\ptr\size-C]$. Thus at the end of iteration $i$, the list $F$
correctly contains FBs in the range  $[b\ptr\rit-C+1,n']$. Similarly
in case~(ii) there can not be any new secondarily forbidden
vertex $[b\ptr\rit-C+1, i\ptr\rit-(f+s+1)-C]$ and hence at the end
of iteration $i$, the list $F$ correctly contains FBs in the range
$[b\ptr\rit-C+1,n']$. At the end of iteration new value of $i$ is
$b$. Hence our claim is true at the beginning of the next iteration too.

Note that for each FB the algorithm takes $O(1)$ time except the operation
$\inlay$. Note that since the operation $\inlay$ is invoked with FBs in
right to left order, it can be implemented by maintaining an extra pointer
that traverses through the FBs in right to left order, in overall $O(|F|)$
time. If the weight of each vertex is at most $C$,\footnote{If the weights
are unrestricted, then we can still solve the decision version of the
problem in $O(n^2)$ time by considering only the interesting FBs that fall
within the set $H(v)$ for a vertex $v$ in $G$. We omit the details here.}
then the set of vertices $H(v)$ for $v$ in  $G$ can contain at most one
follower of each leader FB. There can be as many leader FBs as the number
of maximal cliques in $G'$, \ie at most $n$. Hence $|F|=O(n^2)$. Thus,
time complexity is $O(n^2)$.

Now we show that there is a class of PIGs for which \splitm takes
$\Omega(n^2)$ time. The class of PIGs is obtained by varying some
parameter $t$. A PIG $G$ in this class has $n=3t$ vertices given by
the intervals $I_1,\ldots,I_{3t}$ where for $1 \le j \le t+1$ the
interval $I_j=[j,2t+2i-1]$ has weight $2$, for $2 \le j \le t$ the
interval $I_{t+j}=[t+1,4t+j]$ has weight $2t$, and again for $1 \le j
\le t$ the interval $I_{2t+j}=[2t+2i,5t+j]$ has weight $2$. Clearly
$G'=WXP(G)$ has $(t+1)*2+(t-1)*2t+t*2=2(t^2+t+1)$ vertices and each
of the $t+1$ maximal cliques $[2j-1,2j+2t^2]$, $1 \le j \le t+1$,
has size $2t^2+2$. For $C=2t$, we have $k(G')=t$. In phase~1 \splitm
creates $t+1$ FBs each having a single vertex $2j+2t^2-1$ for all $1 \le
j \le t+1$. In phase~2 \splitm creates a FB from each of the remaining
odd numbered vertices in $G'$. Thus the total number of FBs created by
\splitm on $G$ is equal to the number of odd vertices in $G'$. Hence
$|F|=t^2+t+1=(n/3)^2+n/3+1=\Omega(n^2)$. Thus \splitm takes $\Omega(n^2)$
time on~$G$.
\end{proof}

\subsection{Algorithm \splitp}

We now give the modifications to \combp to use the FBs. We call this modified algorithm \splitp,
which is shown in Algorithm~\ref{alg:splitp}.

\begin{algorithm}[h!tb]
 \SetKwInOut{input}{Input}
 \SetKwInOut{output}{Output}

 \input{A PIG $G=(V,E,W)$}
 \output{If $G'{=}WXP(G)$ has a $[k{+}1, C]$-partition; if \yes also output the partition}

 $F=$ list of FBs returned by \splitm on $G$; \quad $u=1$; \quad $i=F.\bg\ptr\nxt$\;
 \While{$u\le n'$}{
 $v=\min\{u+C-1, n'\}$\;
   \lWhile{$i\ptr\rit < v$}{$i = i\ptr\nxt$}\;
   $v=i\ptr\rit - i\ptr\size$\;
   \lIf{$v<u$}{\Return{\no}} \lElse{create part $[u,v]$; $u=v+1$}\;
 }
 \Return{\yes}\;
\caption{\splitp}
\label{alg:splitp}
\end{algorithm}

\begin{lemma}
If an interval representation for a weighted PIG $G$ is given then \splitp
correctly solves the partition subproblem on $WXP(G)$ in $O(n^2)$ time.
\label{lem:splitpart}
\end{lemma}

\begin{proof}
Given that \splitm correctly marks the forbidden blocks of $G'=WXP(G)$,
it is easy to see that \splitp generates the same partition that \combp
would have generated on $G'$. Given an interval representation of $G$,
$Z,\lmn$ and maximal cliques of $G$ can be computed in $O(n)$ time. Thus
by Lemma~\ref{lem:split_mark}, computing $F$ takes $O(n^2)$ time. The
block generation step also takes $O(|F|)=O(n^2)$ time.
\end{proof}

Combining
Lemmas~\ref{lem:color_using_part},~\ref{lem:wtexpand},~\ref{lem:splitpart}
and the discussions at the end of Subsection~\ref{ssec:ub}, we get a
proof of Theorem~\ref{thm:splittable}. Note that \simp can be slightly
modified to use vertices in $WXP(G)$ but still taking $O(n)$ time.

\section{A $2$-approximation Algorithm for Weighted Problem on PIGs}
\label{sec:weighted}

\begin{lemma}
There exists a polynomial time algorithm for the non-splittable weighted
partition problem that generates a $[\lambda,C]$-partition on a PIG $G$
such that $\lambda$ is at most $2$ times the clique intersection of the
partition generated by an optimal algorithm on $G$.
\label{lem:weighted}
\end{lemma}

\begin{proof}
We first solve the corresponding splittable weighted problem
on $G$ in $O(n^2)$ time using the algorithm described
in Section~\ref{sec:splittable}. Let the blocks in the
$[\lambda',C]$-partition created by the algorithm be $\mathcal{P}'=\{P_1,
P_2, \ldots, P_t\}$. Note that $\lambda'$ is a lower bound on the
clique intersection $\lambda^{*}$ of the partition generated by any
optimal algorithm on $G$.

Since the weight of a vertex in a non-splittable problem is at most $C$,
a vertex of $G$ is split into at most two consecutive blocks $P_i$
and $P_{i+1}$. We convert the splittable partition $\mathcal{P}'$
into a non-splittable partition $\mathcal{P}$ in $O(n)$ time as
follows. Consider each vertex $v$ left to right. If $v$ is split into
blocks $P_i, P_{i+1}$ and $v$ cannot be put completely in $P_{i}$
then create a copy $P'_i$ of $P_i$, insert $P'_i$ in between $P_i$ and
$P_{i+1}$, put $v$ completely in $P'_i$ and repeat with the rest of the
vertices. Note that $\mathcal{P}$ contains at most 2 copies of each
block $P_i$ and hence $\mathcal{P}$ is a $[\lambda,C]$-partition with
clique intersection $\lambda \le 2\lambda' \le 2\lambda^{*}$. Overall
it takes $O(n^2)$ time.
\end{proof}

Combining Lemma~\ref{lem:color_using_part} and Lemma~\ref{lem:weighted}
we get a proof of Theorem~\ref{thm:weighted}.

\section{Partition Problem on Split Graphs}
\label{sec:split}

Since split graphs are also chordal, solving the partition (not
block-partition) problem is enough. It can be noted that the same lower
bound of Lemma~\ref{lem:lb} applies here too.

\subsection{Upper Bound}

\begin{lemma}
 Let $\omega$ be the clique number of a split graph $G$. There exists
 a polynomial time algorithm that gives a $[\ceil{\omega/C} + 1,
 C]$-partition for $G$.  \label{lem:split_ub}
\end{lemma}

\begin{proof}
 Let the vertex set of $G$ be split into clique $Q$ and independent
 set $S$.  Without loss of generality, we assume that $Q$ is a maximum
 clique. Because otherwise $|Q|=\omega-1$ and we can move a vertex in
 $S$ that is adjacent to all vertices in $Q$ to $Q$. Now consider the
 following partition of vertices: $\Pi=\{P_1, P_2, \ldots, P_t\} \cup
 \{\{v\} | v \in S\}$ where $t = \ceil{\omega/C}$, and $\{P_i\}_{i=1}^t$
 is an arbitrary partition of $Q$ such that $|P_i| = C$ for all $i=1,
 \ldots, (t-1)$. Note that this partition can be created in polynomial
 time. Each part is connected and has at most $C$ vertices. Moreover,
 any maximal clique in $G$ intersects at most $t+1$ parts. Thus, $\Pi$
 is a $[\ceil{\omega/C}+1,C]$- partition.
\end{proof}

\subsection{NP-hardness}

Since the upper bound and the lower bound differ by $1$, it is enough
to decide if $G$ has a $[\ceil{\omega/C},C]$-partition or not. If
the answer is \yes then we have an optimal solution to the partition
problem with clique intersection $\lambda=\ceil{\omega/C}$. Otherwise
the partition given in the proof of Lemma~\ref{lem:split_ub}
gives an optimal solution with clique intersection $\lambda =
\ceil{\omega/C}+1$. Thus Lemma~\ref{lem:split_npc} directly gives a
proof of Theorem~\ref{thm:split}.

\begin{lemma}
 The problem of deciding if a split graph $G$ has a
 $[\ceil{\omega(G)/C},C]$-partition for $C\ge 2$ is NP-complete.
 \label{lem:split_npc}
\end{lemma}

\begin{proof}
 We show that the decision problem is NP-complete even for $C=2$. We
 call the problem in this special case as CP. First we show that
 CP is in NP. A maximal clique in $G$ is either $Q$ or the closed
 neighborhood of a vertex in $S$. So the maximal cliques in $G$ can
 be found in polynomial time. Suppose a partition of the vertices is
 given. Size constraints can be easily checked. Each part contains a
 single vertex or a pair of vertices. A single vertex is trivially
 connected. Connectedness of a part of size 2 can be checked by
 just checking if there is an edge between the two vertices. Clique
 intersection constraint can also be checked in polynomial time.

 We now introduce a set partitioning problem (SP) is defined as
 follows. Given a set of $2n$ elements $e_1, e_2, \ldots, e_{2n}$ and
 a collection of $m$ subsets $S_1, S_2, \ldots, S_m$, can the elements
 be partitioned into $n$ groups of size $2$ such that each subset has
 both elements of at least one group?

 We complete the proof by first showing a polynomial time reduction
 from SP to CP (Lemma~\ref{lem:sp2cp}) and then a polynomial
 time reduction from the well known NP-complete problem SAT to SP
 (Lemma~\ref{lem:sat2sp}).
\end{proof}

\begin{lemma}
 SP $\le_P$ CP.
 \label{lem:sp2cp}
\end{lemma}

\begin{proof}
 Given an instance of SP, we construct an instance of CP as follows. The
 complete set $Q$ has a vertex $v_i$ corresponding to each element $e_i$
 and the independent set $S$ has a vertex $w_j$ corresponding to each
 subset $S_j$. There is an edge between $v_i$ and $w_j$ if and only if
 $e_i \notin S_j$. Clearly $\omega = 2n$ and hence $\ceil{\omega/C} = n$.

 Suppose there is a \yes solution to the SP instance where the groups
 are $G_1, G_2, \ldots, G_n$. Then create a partition $\Pi = \{P_1, P_2,
 \ldots, P_n\} \cup \{w_j \}_{j=1}^{m}$ for the CP instance where $v_k
 \in P_i$ if and only if $e_k \in G_i$.  Clearly each part is connected
 and has at most 2 vertices. Clique intersection constraint is satisfied
 for $Q$. Since $S_j$ contains both elements of at least one $G_i$,
 the maximal clique $Q'$ containing $w_j$ does not intersect at least
 one part $P_i$. Including the part $\{w_j\}$, $Q'$ intersects at most
 $(n-1)+1=n$ parts. Hence $\Pi$ is a $[n,2]$-partition.

 On the other hand, suppose there is a \yes solution to the CP instance.
 Since the clique intersection constraint is satisfied for $Q$, the
 vertices of $Q$ are divided into parts of size exactly $2$. These
 parts give the required groups of SP because, for a maximal clique $Q'$
 containing $w_j$ has clique intersection at most $n$, and hence it must
 not intersect with at least one part $P_i$ which implies that $S_j$
 contains both elements of $G_i$.
\end{proof}

\begin{lemma}
 SAT $\le_P$ SP.
 \label{lem:sat2sp}
\end{lemma}

\begin{proof}
 Suppose an instance of SAT has $p$ Boolean variables $x_1, x_2, \ldots,
 x_p$ and $q$ clauses $C_1, C_2, \ldots, C_q$. Without loss of generality,
 we assume that there is at most one literal for each variable in
 each clause.  Now we construct an instance of SP as follows. There
 are $4p$ elements $x_1, x'_1, T_1, F_1,x_2, x'_2, T_2, F_2, \ldots,
 x_p, x'_p, T_p, F_p$ and the subsets are of two types as follows: (1)
 the subsets $\{x_i, x'_i, T_i\}, \{x_i, x'_i, F_i\}, \{x_i, T_i, F_i\},
 \{x'_i,T_i,F_i\}$ for all $1 \le i \le p$, and (2) the subset $\cup_{l_i
 \in C_j} \{l_i, T_i\}$ for all clause $C_j$, where $l_i$ is either $x_i$
 or $x'_i$ (\eg for $C_j = (x_1 + x_2' + x_3)$ the subset $\{x_1, T_1,
 x'_2, T_2, x_3, T_3\}$).

 Suppose there is a satisfying assignment for the SAT instance. Then
 construct a grouping for the SP instance as follows. For all $i$, if
 $x_i$ is true then construct two groups $\{x_i, T_i\}$ and $\{x'_i,
 F_i\}$; otherwise (\ie $x_i$ is false) construct two groups $\{x_i,
 F_i\}$ and $\{x'_i, T_i\}$. Clearly each subset of type (1) has two
 elements belonging to the same group. Since each clause is satisfied
 there must be a variable $x_i$ such that one of $x_i$ and $x'_i$
 is true. Hence the corresponding subset of type (2) must have two
 elements belonging to the same group.

 On the other hand, suppose there is a \yes solution for the SP
 instance. The subsets of type (1) force the elements $x_i, x'_i,
 T_i, F_i$ to form 2 groups amongst themselves. The subsets of type
 (2) ensure that one of the literals $l_i$ in the clause $C_j$ must
 group with $T_i$ and hence $C_j$ must be true. This implies that the
 SAT instance has an satisfying assignment implied by the groups. For
 some $x_i$ the grouping may contain $\{x_i, x'_i\}, \{T_i, F_i\}$, in
 which case the value for the variable $x_i$ can be chosen arbitrarily.
\end{proof}

\section{Conclusions and Future Work}
\label{sec:conclusion}

We gave polynomial time algorithms for unweighted and splittable weighted
versions of the component coloring problem for proper interval graphs
and showed that it is NP-hard for split graphs.  However the complexity
of both the versions are not known for general interval graphs. We would
like to get polynomial time algorithms for general interval graphs
using similar ideas. This may lead to a constant factor approximation
algorithm for the weighted version of the problem for general interval
graphs which is known to be NP-hard, using ideas from Bin-packing.

\bibliographystyle{elsarticle-num}
\biboptions{sort}
\bibliography{lighttrail,ip,graph}

\begin{thebibliography}{10}
\expandafter\ifx\csname url\endcsname\relax
  \def\url#1{\texttt{#1}}\fi
\expandafter\ifx\csname urlprefix\endcsname\relax\def\urlprefix{URL }\fi
\expandafter\ifx\csname href\endcsname\relax
  \def\href#1#2{#2} \def\path#1{#1}\fi

\bibitem{garey1979cig}
M.~R. Garey, D.~S. Johnson, Computers and Intractability: A Guide to the Theory
  of NP-Completeness, W. H. Freeman \& Co., New York, NY, USA, 1979.

\bibitem{chlamtac2003light}
I.~Chlamtac, A.~Gumaste, {Light-trails: A Solution to IP Centric Communication
  in the Optical Domain}, Lecture Notes in Computer Science (2003) 634--644.

\bibitem{pal2009slt}
S.~Pal, A.~Ranade, {Scheduling Light-trails on WDM Rings}, in: {Proceedings of
  the 17th International Conference on Advanced Computing and Communications
  (ADCOM)}, Advanced Computing and Communications Society, 2009, pp. 227--234.

\bibitem{elgindy1996rmb}
H.~ElGindy, H.~Schroder, A.~Spray, A.~Somani, H.~Schmeck, {RMB - A
  reconfigurable multiple bus network}, in: {Second International Symposium on
  High-Performance Computer Architecture}, IEEE, 1996, pp. 108--117.

\bibitem{DBLP:journals/ijcsa/WankarA09}
R.~Wankar, R.~Akerkar, {Reconfigurable Architectures and Algorithms: A Research
  Survey}, IJCSA 6~(1) (2009) 108--123.

\bibitem{olariu1991optimal}
S.~Olariu, {An Optimal Greedy Heuristic to Color Interval Graphs}, Information
  Processing Letters 37~(1) (1991) 21--25.

\bibitem{golumbic2004algorithmic}
M.~C. Golumbic, Algorithmic Graph Theory and Perfect Graphs (Annals of Discrete
  Mathematics, Vol 57), North-Holland Publishing Co., Amsterdam, The
  Netherlands, 2004.

\bibitem{edwards2005monochromatic}
K.~Edwards, G.~Farr, {On Monochromatic Component Size for Improper Colourings},
  Discrete Applied Mathematics 148~(1) (2005) 89--105.

\bibitem{linial2008graph}
N.~Linial, J.~Matou{\v{S}}ek, O.~Sheffet, G.~Tardos, {Graph Colouring with no
  Large Monochromatic Components}, Combinatorics, Probability and Computing
  17~(04) (2008) 577--589.

\bibitem{balasubramanian2005ltn}
S.~Balasubramanian, W.~He, A.~Somani, {Light-Trail networks: design and
  survivability}, Thirtieth IEEE Conference on Local Computer Networks (2005)
  174--181.

\bibitem{ayad2007eoa}
A.~Ayad, K.~Elsayed, S.~Ahmed, {Enhanced optimal and heuristic solutions of the
  routing problem in Light-trail networks}, Workshop on High Performance
  Switching and Routing (HPSR) (2007) 1--6.

\bibitem{gumaste2007hao}
A.~Gumaste, P.~Palacharla, {Heuristic and optimal techniques for light-trail
  assignment in optical ring WDM networks}, Computer Communications 30~(5)
  (2007) 990--998.

\bibitem{wu2006opn03}
B.~Wu, K.~Yeung, {OPN03-5: Light-trail Assignment in WDM Optical Networks}, in:
  {IEEE Global Telecommunications Conference (GLOBECOM)}, 2006, pp. 1--5.

\bibitem{luo2009integrated}
X.~Luo, B.~Wang, {Integrated scheduling of grid applications in WDM optical
  Light-trail networks}, Journal of Lightwave Technology 27~(12) (2009)
  1785--1795.

\bibitem{gokhale2010cloud2}
P.~Gokhale, R.~Kumar, T.~Das, A.~Gumaste, {Cloud computing over Metropolitan
  Area WDM networks: The Light-trails approach}, in: IEEE Global
  Telecommunications Conference (GLOBECOM), IEEE, 2010, pp. 1--6.

\bibitem{looges1993optimal}
P.~Looges, S.~Olariu, {Optimal Greedy Algorithms for Indifference Graphs},
  Computers \& Mathematics with Applications 25~(7) (1993) 15--25.

\end{thebibliography}

\end{document}